\documentclass{amsart}

\usepackage{hyperref}

\newtheorem{theorem}{Theorem}[section]

\newtheorem{lemma}[theorem]{Lemma}

\theoremstyle{definition}

\newtheorem{remark}[theorem]{Remark}

\numberwithin{equation}{section}

\DeclareMathOperator{\supp}{supp}
\DeclareMathOperator{\dist}{dist}

\begin{document}
\title[Quantitative Unique Continuation Principle]{Quantitative unique continuation principle for Schr\"{o}dinger Operators with Singular Potentials}
\author{Abel Klein \and C.S. Sidney Tsang}
\address{University of California, Irvine,
Department of Mathematics,
Irvine, CA 92697-3875,  USA}
 \email[Klein]{aklein@uci.edu}
  \email[Tsang]{tsangcs@uci.edu}

  \thanks{A.K. and C.S.S.T. were  supported  by the NSF under grant DMS-1301641.}


\begin{abstract} We prove  a quantitative unique continuation principle for Schr\"{o}dinger operators  $H=-\Delta+V$ on $\mathrm{L}^2(\Omega)$, where $\Omega$ is an open subset of  $\mathbb{R}^d$  and  $V$ is a singular potential:
$V \in \mathrm{L}^\infty(\Omega) + \mathrm{L}^p(\Omega)$.   As an application, we derive a unique continuation principle for spectral projections  of  Schr\"{o}dinger operators with singular potentials.
\end{abstract}

\maketitle

\section{Introduction}

We prove a quantitative unique continuation principle for Schr\"{o}dinger operators $H=-\Delta+V$ on $\mathrm{L}^2(\Omega)$, where $\Omega$ is an open subset of  $\mathbb{R}^d$, $\Delta$  is the Laplacian operator, and  $V$ is a singular real  potential:
$V \in \mathrm{L}^\infty(\Omega) + \mathrm{L}^p(\Omega)$.  Our results extend the original result of Bourgain and Kenig \cite[Lemma~3.10]{BK}, as well as  subsequent versions  \cite[Theorem~A.1]{GKloc} and \cite[Theorem~3.4]{BKl},  where $V$ is a bounded potential: $V \in \mathrm{L}^\infty(\Omega)$.

As an application, we derive a unique continuation principle for spectral projections  of  Schr\"{o}dinger operators with singular potentials, extending the bounded potential results of   \cite[Theorem 1.1]{Kl} and  \cite[Theorem B.1]{KN}.

To prove the quantitative unique continuation principle for singular potentials we use Sobolev inequalities (not required for bounded potentials).
Since the Sobolev inequality we use in dimension $d=2$ is expressed in terms of Orlicz norms, we review Orlicz spaces, following \cite{RR}.
A function\ $\varphi:\mathbb{R}^+\rightarrow\mathbb{R}^+\cup\{+\infty\}$\ is called a Young function if it is increasing, convex,\ $\varphi(0)=0$, and  $\lim_{t\to\infty}\varphi(t)=\infty$.  Its  complementary function,  given by $\varphi^{\ast}(t)=\sup_{s\in\mathbb{R}^+}\{st-\varphi(s)\}$ for $t\in\mathbb{R}^+$,   is also a Young function. Given  a Young function $\varphi$  and a $\sigma$-finite  measure $\mu$ on a measurable space  $X$, we define the Orlicz space
\begin{equation}
\mathrm{L}^{\varphi}(X)=\left\{f:X\rightarrow\mathbb{R}\ \mbox{measurable}\left|\int_X\varphi(\alpha |f|)d\mu<\infty\ \mbox{for some}\ \alpha>0\right.\right\},
\end{equation}
a Banach space  when
equipped with the Orlicz norm
\begin{equation}
\|f\|_{\varphi}:=\inf\left\{k>0:\int_X\varphi\left(\tfrac{1}{k}|f|\right)d\mu\leq1\right\}.
\end{equation}
  (A standard example is $\varphi(t)= t^p$ with  $1\leq p<\infty$; in this case $\mathrm{L}^{\varphi}(X)= \mathrm{L}^{p}(X)$.)   There is a H\"{o}lder's inequality for Orlicz spaces:
\begin{equation}\label{osholder}
\int_X|fg|d\mu\leq2\|f\|_{\varphi}\|g\|_{\varphi^{\ast}} \quad\text{for all } \quad f\in\mathrm{L}^{\varphi}(X), \; g\in\mathrm{L}^{\varphi^{\ast}}(X).
\end{equation}

We now state our main theorem,  a quantitative unique continuation principle for Schr\"{o}dinger operators with singular potentials. We fix  the Young function
\begin{equation}
\varphi(t)=e^t-1, \quad\text{so}\quad \varphi^{\ast}(t)=\left\{\begin{array}{ll}0&\mbox{if}\quad 0\leq t\leq1
\\t\log t-t+1&\mbox{if}\quad t>1\end{array}\right..
\end{equation}
We use the norm $|x|:=(\sum_{j=1}^d|x_j|^2)^{\frac{1}{2}}$  for $x=(x_1,x_2,\ldots,x_d)\in\mathbb{R}^d$; all distances in\ $\mathbb{R}^d$\ will be measured with respect to this norm.  By $B(x,\delta):=\{y\in\mathbb{R}^d:|y-x|<\delta\}$ we denote the ball centered at $x\in\mathbb{R}^d$ with radius $\delta>0$.  Given subsets\ $A$\ and\ $B$\ of\ $\mathbb{R}^d$, and a function\ $\phi$ on set  $B$, we set  $\phi_A:=\phi\chi_{A\cap B}$. We let   $\phi_{x,\delta}:=\phi_{B(x,\delta)}$.

\begin{theorem}\label{qucp3d}
Let\ $\Omega$ be an open subset of $\mathbb{R}^d$,   $K=K_1+K_2$ with $K_1, K_2 \ge 0$,  and consider a real measurable function\ $V=V^{(1)}+V^{(2)}$ on\ $\Omega$\ with $\|V^{(1)}\|_{\infty}\leq K_1$. Let $\psi\in\mathrm{L}^2(\Omega)$ be real valued with $\Delta \psi\in\mathrm{L}_{loc}^2(\Omega)$, and suppose
\begin{equation}\label{q3d1}
\zeta=-\Delta\psi+V\psi\in\mathrm{L}^2(\Omega).
\end{equation}
Fix a bounded measurable set $\Theta\subset\Omega$ where $\|\psi_{\Theta}\|_{2}>0$, and set
\begin{equation}\label{q3d2}
Q(x,\Theta):=\sup_{y\in\Theta}|y-x|\quad\mbox{for}\quad x\in\Omega.
\end{equation}
Consider\ $x_0\in\Omega\backslash\overline{\Theta}$ such that
\begin{equation}\label{q3d3}
Q=Q(x_0,\Theta)\geq1\quad \mbox{and}\quad B(x_0,6Q+2)\subset\Omega,
\end{equation}
and take
\begin{equation}\label{q3d4}
0<\delta\leq\min\{\dist(x_0,\Theta),\tfrac{1}{2}\}.
\end{equation}
 There  is a constant $m_d>0$,  depending only on $d$, such that:

\begin{enumerate}
\item[(i)]  If  either $d\geq3$ and $\|V^{(2)}\|_p\leq K_2$ with  $p\geq d$, or
$d=2$ and     $(\||V^{(2)}|^p\|_{\varphi^{\ast}})^{\frac{1}{p}}\leq K_2$ with  $p\geq2$,
we have
\begin{equation}\label{q3d5}
\left(\frac{\delta}{Q}\right)^{m_d(1+K^{\frac{2p}{3p-2d}})(Q^{\frac{4p-2d}{3p-2d}}+\log\frac{\|\psi_{\Omega}\|_{2}}{\|\psi_{\Theta}\|_{2}})}\|\psi_{\Theta}\|_{2}^2\leq\|\psi_{x_0,\delta}\|_{2}^2+\delta^2\|\zeta_{\Omega}\|_{2}^2.
\end{equation}
  In particular,  if $d=2$ it suffices to require $\|V^{(2)}\|_p\leq K_2$  with $p>2$ to obtain \eqref{q3d5}.
\item[(ii)] If $d=1$ and $\|V^{(2)}\|_p\leq K_2$ with  $p\geq2$,  we have
\begin{equation}\label{q3d6}
\left(\frac{\delta}{Q}\right)^{m_1(1+K^{\frac{2p}{3p-4}})(Q^{\frac{4p-4}{3p-4}}+\log\frac{\|\psi_{\Omega}\|_{2}}{\|\psi_{\Theta}\|_{2}})}\|\psi_{\Theta}\|_{2}^2\leq\|\psi_{x_0,\delta}\|_{2}^2+\delta^2\|\zeta_{\Omega}\|_{2}^2.
\end{equation}

\end{enumerate}
 \end{theorem}

Letting  $p\to \infty$ in Theorem~\ref{qucp3d} we recover \cite[Theorem~3.4]{BKl}.
The proof of Theorem~\ref{qucp3d}, given  in Section~\ref{secproofthm},
 relies on a Carleman estimate of Escauriaza and Vesella \cite[Theorem~2]{EV}, stated in Lemma~\ref{EVineq}.  To control singular potentials we use  all the terms  in this estimate, including the the gradient term, and Sobolev's inequalities.  In the
 proofs for bounded potentials \cite{BK,GKloc,BKl} it suffices to use a simpler version of this Carleman estimate without the the gradient term (see \cite[Lemma~3.15]{BK}).

As an application of Theorem~\ref{qucp3d}, we prove a unique continuation principle for spectral projections  of Schr\"{o}dinger operators with singular potentials, extending   \cite[Theorem 1.1]{Kl} (in the form given in  \cite[Theorem B.1]{KN}) to Schr\"{o}dinger operators with singular potentials. (See also
 \cite[Section~4]{CHK1}, \cite[Theorem~2.1]{CHK2}, \cite[Theorem~A.6]{GKloc}, and \cite[Theorem~2.1]{RV} for unique continuation principles for spectral projections  of Schr\"{o}dinger operators with bounded potentials.)

We consider rectangles in $\mathbb{R}^d$  of the form
\begin{equation}\label{ucp2}
\Lambda=\Lambda_{\mathbf{L}}(a)=a+\prod_{j=1}^d\left(-\tfrac{L_j}{2},\tfrac{L_j}{2}\right)=\prod_{j=1}^d\left(a_j-\tfrac{L_j}{2},a_j+\tfrac{L_j}{2}\right),
\end{equation}
where $a=(a_1,\ldots,a_d)\in\mathbb{R}^d$ and $\mathbf{L}=(L_1,\ldots,L_d)\in(0,\infty)^d$.
 (We write $\Lambda_L(a)=\Lambda_{\mathbf{L}}(a)$ in the special case $L_j=L$ for $j=1,\ldots,d$.) Given  a Schr\"{o}dinger operator $H=-\Delta+V$  on $\mathrm{L}^2(\mathbb{R}^d)$,
by $H_{\Lambda}=-\Delta_{\Lambda}+V_{\Lambda}$ we denote the restriction of $H$ to the rectangle   $\Lambda$ with either Dirichlet or periodic boundary condition: $\Delta_{\Lambda}$ is the Laplacian on $\Lambda$  with either Dirichlet or periodic boundary condition, and $V_{\Lambda}$ is the restriction of $V$ to $\Lambda$.

\begin{theorem}\label{ucpsp}
Let $H=-\Delta+V$ be a Schr\"{o}dinger operator on $\mathrm{L}^2(\mathbb{R}^d)$, where $V=V^{(1)}+V^{(2)}$ with $\|V^{(1)}\|_{\infty}\leq K_1<\infty$ and $\|V^{(2)}\|_p\leq K_2<\infty$ with $p\geq d$ for $d\geq3$, $p>2$ for $d=2$, and $p\geq2$ for $d=1$.   Set $K=K_1+K_2$.  Fix $\delta\in(0,\frac{1}{2}]$, and let $\{y_k\}_{k\in\mathbb{Z}^d}$ be sites in $\mathbb{R}^d$ with $B(y_k,\delta)\subset\Lambda_1(k)$\ for all $k\in\mathbb{Z}^d$. There exists a constant $M_d>0$, depending only on $d$,  such that, defining $\gamma = \gamma(d,p,K,\delta,E_0) > 0$ for $E_0>0$  by
\begin{equation}\label{ucpsp2}
\gamma^2=\left\{\begin{array}{ll}
\frac{1}{2}\delta^{M_d\left(1+(K+E_0)^{\frac{4p^2}{(3p-2d)(2p-d)}}\right)}&\mbox{for}\quad d\geq2\\\frac{1}{2}\delta^{M_d\left(1+(K+E_0)^{\frac{2p^2}{(3p-4)(p-1)}}\right)}&\mbox{for}\quad d=1
\end{array}
\right.,
\end{equation}
then,
 given a rectangle $\Lambda$ as in \eqref{ucp2}, where $a\in\mathbb{R}^d$ and $L_j\geq114\sqrt{d}$ for $j=1,\ldots,d$, and a  closed interval $I\subset(-\infty,E_0]$ with $|I|\leq2\gamma$, we have
 \begin{equation}\label{ucpsp3}
\chi_I(H_{\Lambda})W^{(\Lambda)}\chi_I(H_{\Lambda})\geq\gamma^2\chi_I(H_{\Lambda}),
\end{equation}
where
\begin{equation}\label{ucpsp1}
W^{(\Lambda)}=\sum_{k\in \mathbb{Z}^d, \, \Lambda_1(k)\subset\Lambda}\chi_{B(y_k,\delta)}.
\end{equation}
\end{theorem}

The proof of Theorem~\ref{ucpsp} is discussed in Section~\ref{secUCPSP}.

\begin{remark} Using Theorem~\ref{ucpsp} we can prove optimal Wegner estimates for Anderson Hamiltonians with singular background potentials, extending the results  of \cite{Kl}.
\end{remark}

\section{Quantitative unique continuation principle for Schr\"{o}dinger operators with singular potentials}\label{secproofthm}

 The proof of Theorem \ref{qucp3d} is based on a Carleman estimate  of Escauriaza and Vesella \cite[Theorem~2]{EV}, which we state in a ball of radius $\varrho>0$.
\begin{lemma}\label{EVineq}
Given\ $\varrho>0$, the function $\omega_{\varrho}(x)=\phi(\frac{1}{\varrho}|x|)$ on\ $\mathbb{R}^d$, where\ $\phi(s):=se^{-\int_0^s\frac{1-e^{-t}}{t}dt}$, is a strictly increasing continuous function on\  $[0,\infty)$,\ $\mathrm{C}^{\infty}$\ on\ $(0,\infty)$, satisfying
\begin{equation}\label{ei1}
\frac{1}{C_1\varrho}|x|\leq\omega_{\varrho}(x)\leq\frac{1}{\varrho}|x|\quad \mbox{for}\quad x\in B(0,\varrho),
\end{equation}
where\ $C_1=\phi(1)^{-1}\in(2,3)$. Moreover, there exist positive contants\ $C_2$\ and\ $C_3$, depending only on $d$, such that for all $\alpha\geq C_2$ and all real valued functions\ $f\in H^2(B(0,\varrho))$ with $\supp f\subset B(0,\varrho)\backslash\{0\}$ we have
\begin{equation}\label{ei2}
\alpha^3\int_{\mathbb{R}^d}\omega_{\varrho}^{-1-2\alpha}f^2dx+\alpha\varrho^2\int_{\mathbb{R}^d}\omega_{\varrho}^{1-2\alpha}|\nabla f|^2dx\leq C_3\varrho^4\int_{\mathbb{R}^d}\omega_{\varrho}^{2-2\alpha}(\Delta f)^2dx.
\end{equation}
\end{lemma}

This estimate  is given in the parabolic setting in  \cite{EV}, but the estimate in the elliptic setting as in the lemma follows immediately by the argument in \cite[Proposition B.3]{KSU}. In the proofs of the quantitative unique continuation principle for bounded potentials \cite{BK,GKloc,BKl} only the first term in the left hand side of\ \eqref{ei2} is used (see \cite[Lemma~3.15]{BK}), but for singular potentials we also need to use the gradient term in the left hand side of\ \eqref{ei2} and Sobolev's inequalities.

\begin{proof}[Proof of Theorem \ref{qucp3d}]    Let $C_1,C_2,C_3$ be the constants of Lemma~ \ref{EVineq}, which depend only on $d$.  Without loss of generality $C_2>1$. By $C_j$,\ $j=4,5,\ldots$, we will always denote an appropriate nonzero constant depending only on $d$.

We follow Bourgain and Klein's proof for bounded potentials \cite[Theorem~3.4]{BKl}.
Let $x_0\in\Omega\backslash\overline{\Theta}$ be as in  \eqref{q3d3}.  Without loss of generality we take $x_0=0$, $\Theta\subset B(0,2C_1Q)$, and $\Omega=B(0,\varrho)$, where\ $\varrho=2C_1Q+2$, and let $\delta$ be as in \eqref{q3d4}.
Proceeding as in \cite[Theorem~3.4]{BKl}, we
 fix a function $\eta\in\mathrm{C}_c^{\infty}(\mathbb{R}^d)$ given by $\eta(x)=\xi(|x|)$, where $\xi$ is an even $\mathrm{C}^{\infty}$ function on $\mathbb{R}$, $0\leq\xi\leq1$, such that
\begin{align}\label{q3dp1}
&\xi(s)=1\quad \mbox{if}\quad \tfrac{3}{4}\delta \leq|s|\leq2C_1Q, \quad
 \xi(s)=0\quad \mbox{if}\quad |s|\leq\tfrac{1}{4}\delta\ \mbox{or}\ |s|\geq2C_1Q+1,
\nonumber\\&|\xi^j(s)|\leq\left(\tfrac{4}{\delta}\right)^j\quad \mbox{if}\quad |s|\leq\tfrac{3}{4}\delta,\quad
|\xi^j(s)|\leq2^j\quad \mbox{if}\quad |s|\geq2C_1Q, j=1,2,\\
\nonumber &  |\nabla\eta(x)|\leq\sqrt{d}|\xi'(|x|)| \quad \text{and}\quad |\Delta\eta(x)|\leq d|\xi''(|x|)|,\\ \nonumber &
\supp\nabla\eta\subset\{\tfrac{\delta}{4}\leq|x|\leq\tfrac{3\delta}{4}\}\cup\{2C_1Q\leq|x|\leq2C_1Q+1\}.
\end{align}

Let  $\alpha\geq C_2$.
Applying Lemma \ref{EVineq} to the function  $\eta\psi$  gives
\begin{align}\notag
&\frac{\alpha^3}{3C_3\varrho^4}\int_{\mathbb{R}^d}\omega_{\varrho}^{-1-2\alpha}\eta^2\psi^2dx+\frac{\alpha}{3C_3\varrho^2}\int_{\mathbb{R}^d}\omega_{\varrho}^{1-2\alpha}|\nabla(\eta\psi)|^2dx \\ \label{q3dp2} & \qquad
\leq \frac{1}{3}\int_{\mathbb{R}^d}\omega_{\varrho}^{2-2\alpha}(\Delta(\eta\psi))^2dx
\leq \int_{\mathbb{R}^d}\omega_{\varrho}^{2-2\alpha}\eta^2(\Delta\psi)^2dx \\
\nonumber &  \qquad \quad
+4\int_{\supp\nabla\eta}\omega_{\varrho}^{2-2\alpha}|\nabla\eta|^2|\nabla\psi|^2dx
+\int_{\supp\nabla\eta}\omega_{\varrho}^{2-2\alpha}(\Delta\eta)^2\psi^2dx.
\end{align}

Using \eqref{q3d1},  $\|V^{(1)}\|_{\infty}\leq K_1$, and  $\omega_{\varrho}\leq1$ on $\supp\eta$, we have
\begin{align} \label{q3dp3}
 &   \quad  \int_{\mathbb{R}^d}\omega_{\varrho}^{2-2\alpha}\eta^2(\Delta\psi)^2dx\leq2\int_{\mathbb{R}^d}V^2\omega_{\varrho}^{2-2\alpha}\eta^2\psi^2dx+2\int_{\mathbb{R}^d}\omega_{\varrho}^{2-2\alpha}\eta^2\zeta^2dx
\\  & \quad \quad \leq 4K_1^2\int_{\mathbb{R}^d}\omega_{\varrho}^{-1-2\alpha}\eta^2\psi^2dx+4\int_{\mathbb{R}^d}(V^{(2)})^2\omega_{\varrho}^{2-2\alpha}\eta^2\psi^2dx+2\int_{\mathbb{R}^d}\omega_{\varrho}^{2-2\alpha}\eta^2\zeta^2dx. \nonumber
\end{align}
Given  $M>0$, we write $V^{(2)}=U_M + V_M$, where  $U_M=V^{(2)}\chi_{\{|V^{(2)}|\leq\sqrt{M}\}}$ and $W_M=V^{(2)}\chi_{\{|V^{(2)}|>\sqrt{M}\}}$. We have
\begin{align}\label{q3dp4}
\int_{\mathbb{R}^d}(V^{(2)})^2\omega_{\varrho}^{2-2\alpha}\eta^2\psi^2dx\leq& M\int_{\mathbb{R}^d}\omega_{\varrho}^{-1-2\alpha}\eta^2\psi^2dx+\int_{\mathbb{R}^d}W_M^2\omega_{\varrho}^{2-2\alpha}\eta^2\psi^2dx.
\end{align}
Combining \eqref{q3dp2}, \eqref{q3dp3} and \eqref{q3dp4}, we have
\begin{align} \label{q3dp4a}
 &\left(\frac{\alpha^3}{3C_3\varrho^4}-4K_1^2-4M\right)\int_{\mathbb{R}^d}\omega_{\varrho}^{-1-2\alpha}\eta^2\psi^2dx+\frac{\alpha}{3C_3\varrho^2}\int_{\mathbb{R}^d}\omega_{\varrho}^{1-2\alpha}|\nabla(\eta\psi)|^2dx \nonumber \\&\qquad\leq 4\int_{\mathbb{R}^d}W_M^2\omega_{\varrho}^{2-2\alpha}\eta^2\psi^2dx+2\int_{\mathbb{R}^d}\omega_{\varrho}^{2-2\alpha}\eta^2\zeta^2dx \\
\nonumber\ &  \qquad \quad
+4\int_{\supp\nabla\eta}\omega_{\varrho}^{2-2\alpha}|\nabla\eta|^2|\nabla\psi|^2dx
+\int_{\supp\nabla\eta}\omega_{\varrho}^{2-2\alpha}(\Delta\eta)^2\psi^2dx.
\end{align}
Note that for $1\le q\le p$ we have
\begin{align}\label{q3dp5}
\|W_M\|_q \le M^{-\frac{p-q}{2q}}\|W_M\|_p^{\frac{p}{q}} \le M^{-\frac{p-q}{2q}}\|V^{(2)}\|_p^{\frac{p}{q}}\leq M^{-\frac{p-q}{2q}}K_2^{\frac{p}{q}}.
\end{align}
 We set $K=K_1+K_2$ with $K_1, K_2 \ge 0$.

We consider three cases:\smallskip

\noindent{(a)}\ $d\geq3$:   Let  $\|V^{(2)}\|_p\leq K_2$ with  $p\geq d$.
Using H\"{o}lder's inequality and  \eqref{q3dp5} with $q=d$,  we get
\begin{align}\label{q3dp6}
&\int_{\mathbb{R}^d}W_M^2\omega_{\varrho}^{2-2\alpha}\eta^2\psi^2dx\leq\|W_M^2\|_{\frac{d}{2}}\|\omega_{\varrho}^{2-2\alpha}\eta^2\psi^2\|_{\frac{d}{d-2}}
\\&=\|W_M\|_d^2\|\omega_{\varrho}^{1-\alpha}\eta\psi\|_{\frac{2d}{d-2}}^2
\leq M^{-\frac{p-d}{d}}K_2^{\frac{2p}{d}}\|\omega_{\varrho}^{1-\alpha}\eta\psi\|_{\frac{2d}{d-2}}^2.
\nonumber
\end{align}
Using Sobolev's inequality (e.g., \cite[Theorem 7.10]{GiT}), we get
\begin{align}\label{q3dp7}
\|\omega_{\varrho}^{1-\alpha}\eta\psi\|_{\frac{2d}{d-2}}^2&\leq C_4\left(\int_{\mathbb{R}^d}|\nabla(\omega_{\varrho}^{1-\alpha}\eta\psi)|^2\right)
\\&\leq2C_4\int_{\mathbb{R}^d}|\nabla\omega_{\varrho}^{1-\alpha}|^2\eta^2\psi^2dx+2C_4\int_{\mathbb{R}^d}\omega_{\varrho}^{1-2\alpha}|\nabla(\eta\psi)|^2dx.\nonumber
\end{align}
Since
\begin{equation}\label{q3dp8}
|\nabla\omega_{\varrho}^{1-\alpha}|^2=(1-\alpha)^2\frac{\omega_{\varrho}^{2-2\alpha}}{|x|^2\exp(\frac{2}{\varrho}|x|)}\leq\frac{\alpha^2}{\varrho^2}\omega_{\varrho}^{-2\alpha},
\end{equation}
we have (recall $\omega_{\varrho}\leq1$ on $\supp\eta$)
\begin{equation}\label{q3dp9}
\int_{\mathbb{R}^d}|\nabla\omega_{\varrho}^{1-\alpha}|^2\eta^2\psi^2dx\leq\frac{\alpha^2}{\varrho^2}\int_{\mathbb{R}^d}\omega_{\varrho}^{-1-2\alpha}\eta^2\psi^2dx.
\end{equation}
Combining \eqref{q3dp4a}, \eqref{q3dp6}, \eqref{q3dp7} and \eqref{q3dp9}, we conclude that
\begin{align}\label{q3dp10}
&\left(\frac{\alpha^3}{3C_3\varrho^4}-4K_1^2-4M-8C_4M^{-\frac{p-d}{d}}K_2^{\frac{2p}{d}}\frac{\alpha^2}{\varrho^2}\right)\int_{\mathbb{R}^d}\omega_{\varrho}^{-1-2\alpha}\eta^2\psi^2dx
\nonumber\\& \hspace{60pt} +\left(\frac{\alpha}{3C_3\varrho^2}-8C_4M^{-\frac{p-d}{d}}K_2^{\frac{2p}{d}}\right)\int_{\mathbb{R}^d}\omega_{\varrho}^{1-2\alpha}|\nabla(\eta\psi)|^2dx
\nonumber\\ &  \quad \le      4\int_{\supp\nabla\eta}\omega_{\varrho}^{2-2\alpha}|\nabla\eta|^2|\nabla\psi|^2dx
+\int_{\supp\nabla\eta}\omega_{\varrho}^{2-2\alpha}(\Delta\eta)^2\psi^2dx\\  \notag & \hspace{60pt}
+2\int_{\supp\eta}\omega_{\varrho}^{2-2\alpha}\eta^2\zeta^2dx .
\end{align}
Assuming $\alpha\geq\varrho$ and setting  $M=K_2^2\alpha^{\frac{2d}{p}}\varrho^{\frac{-2d}{p}}$,  we have
\begin{align}\label{q3dp11}
4K_1^2+4M+8C_4M^{-\frac{p-d}{d}}K_2^{\frac{2p}{d}}\alpha^2\varrho^{-2}& =4K_1^2+4K_2^2(1+2C_4)\alpha^{\frac{2d}{p}}\varrho^{\frac{-2d}{p}}
\nonumber\\&   \le  (4K^2(1+2C_4))\alpha^{\frac{2d}{p}}\varrho^{\frac{-2d}{p}}.
\end{align}
 Taking
\begin{equation}\label{q3dp13}
\alpha\geq C_5(1+K^{\frac{2p}{3p-2d}})\varrho^{\frac{4p-2d}{3p-2d}}\geq C_5(1+K^{\frac{2p}{3p-2d}})\varrho^{\frac{4}{3}},
\end{equation}
we can guarantee that   $\alpha>C_2$,
\begin{equation}\label{q3dp12}
\frac{\alpha^3}{3C_3\varrho^4}\geq3(4K^2(1+2C_4)\alpha^{\frac{2d}{p}}\varrho^{\frac{-2d}{p}}),
\end{equation}
and
\begin{equation}\label{q3dp14}
\frac{\alpha}{3C_3\varrho^2}-8C_4M^{-\frac{p-d}{d}}K_2^{\frac{2p}{d}}\geq0.
\end{equation}

Using \eqref{ei1} and recalling \eqref{q3d2}, we obtain
\begin{equation}\label{q3dp15}
\int_{\mathbb{R}^d}\omega_{\varrho}^{-1-2\alpha}\eta^2\psi^2dx\geq\left(\frac{\varrho}{Q}\right)^{1+2\alpha}\|\psi_{\Theta}\|_2^2\geq(2C_1)^{1+2\alpha}\|\psi_{\Theta}\|_2^2.
\end{equation}
Combining \eqref{q3dp10}, \eqref{q3dp12}, \eqref{q3dp14} and \eqref{q3dp15}, we conclude that
\begin{align}\label{q3dp16}
&\frac{2\alpha^3}{9C_3\varrho^4}(2C_1)^{1+2\alpha}\|\psi_{\Theta}\|_2^2    \leq4\int_{\supp\nabla\eta}\omega_{\varrho}^{2-2\alpha}|\nabla\eta|^2|\nabla\psi|^2dx
\\ & \hspace{40pt}  +\int_{\supp\nabla\eta}\omega_{\varrho}^{2-2\alpha}(\Delta\eta)^2\psi^2dx+2\int_{\supp\eta}\omega_{\varrho}^{2-2\alpha}\eta^2\zeta^2dx.  \nonumber
\end{align}

Let $f\in\mathcal{D}(\nabla)$. For arbitrary $M>0$ we have
\begin{align}\label{q3dp17}
\left|\int_{\mathbb{R}^d}Vf^2dx\right| \le (K_1+M^{\frac{1}{2}})\|f\|_2^2+\int_{\mathbb{R}^d}|W_M|f^2dx.
\end{align}
Using H\"{o}lder's inequality,  \eqref{q3dp5} with $q=\frac{d}{2}$,
  and Sobolev's inequality, we get
\begin{equation}\label{q3dp19}
\left|\int_{\mathbb{R}^d}Vf^2dx\right|\leq(K_1+M^{\frac{1}{2}})\|f\|_2^2+C_4 M^{-\frac{2p-d}{2d}}K_2^{\frac{2p}{d}}\|\nabla f\|_2^2.
\end{equation}
Taking $M=(2C_4K_2^{\frac{2p}{d}})^{\frac{2d}{2p-d}}$ (we can require $C_4\geq1$), we get
\begin{equation}\label{q3dp20}
\left|\int_{\mathbb{R}^d}Vf^2dx\right|\leq 2C_4(1+K^{\frac{2p}{2p-d}})\|f\|_2^2+\frac{1}{2}\|\nabla f\|_2^2.
\end{equation}

We have
\begin{align}\label{q3dp21}
&\int_{\{2C_1Q\leq|x|\leq2C_1Q+1\}}\omega_{\varrho}^{2-2\alpha}(4|\nabla\eta|^2|\nabla\psi|^2+(\Delta\eta)^2\psi^2)dx
\\& \quad \le16d^2\left(\frac{C_1\varrho}{2C_1Q}\right)^{2\alpha-2}\int_{\{2C_1Q\leq|x|\leq2C_1Q+1\}}(4|\nabla\psi|^2+\psi^2)dx
\nonumber\\ & \quad \le C_6\left(\tfrac{5}{4}C_1\right)^{2\alpha-2}\int_{\{2C_1Q-1\leq|x|\leq2C_1Q+2\}}(\zeta^2+(1+K^{\frac{2p}{2p-d}})\psi^2)dx
\nonumber\\& \quad \le C_6\left(\tfrac{5}{4}C_1\right)^{2\alpha-2}(\|\zeta_{\Omega}\|_2^2+(1+K^{\frac{2p}{2p-d}})\|\psi_{\Omega}\|_2^2),\notag
\end{align}
where we used \eqref{q3dp20} and an interior estimate (e.g., \cite[Lemma A.2]{GK1}).  Similarly,
\begin{align}\label{q3dp22}
&\int_{\{\frac{\delta}{4}\leq|x|\leq\frac{3\delta}{4}\}}\omega_{\varrho}^{2-2\alpha}(4|\nabla\eta|^2|\nabla\psi|^2+(\Delta\eta)^2\psi^2)dx
\\   & \quad \le 256d^2\delta^{-4}(4\delta^{-1}C_1\varrho)^{2\alpha-2}\int_{\{\frac{\delta}{4}\leq|x|\leq\frac{3\delta}{4}\}}(4|\nabla\psi|^2+\psi^2)dx
\nonumber\\& \quad \le  C_7\delta^{-4}(4\delta^{-1}C_1\varrho)^{2\alpha-2}\int_{\{|x|\leq\delta\}}(\zeta^2+(K^{\frac{2p}{2p-d}}+\delta^{-2})\psi^2)dx
\nonumber\\& \quad \le C_7\delta^{-4}(16\delta^{-1}C_1^2Q)^{2\alpha-2}(\|\zeta_{\Omega}\|_2^2+(K^{\frac{2p}{2p-d}}+\delta^{-2})\|\psi_{0,\delta}\|_2^2). \nonumber
\end{align}
In addition,
\begin{equation}\label{q3dp23}
\int_{\supp\eta}\omega_{\varrho}^{2-2\alpha}\eta^2\zeta^2dx\leq(4\delta^{-1}C_1\varrho)^{2\alpha-2}\|\zeta_{\Omega}\|_2^2\leq(16\delta^{-1}C_1^2Q)^{2\alpha-2}\|\zeta_{\Omega}\|_2^2.
\end{equation}

If we have
\begin{equation}\label{q3dp24}
\frac{\alpha^3}{\varrho^4}\left(\frac{8}{5}\right)^{2\alpha}\|\psi_{\Theta}\|_2^2\geq C_8(1+K^{\frac{2p}{2p-d}})\|\psi_{\Omega}\|_2^2,
\end{equation}
we obtain
\begin{equation}\label{q3dp25}
C_6\left(\tfrac{5}{4}C_1\right)^{2\alpha-2}(1+K^{\frac{2p}{2p-d}})\|\psi_{\Omega}\|_2^2\leq\frac{1}{2}\frac{2\alpha^3}{9C_3\varrho^4}(2C_1)^{1+2\alpha}\|\psi_{\Theta}\|_2^2,
\end{equation}
so we conclude that
\begin{align}\label{q3dp26}
&\frac{\alpha^3}{9C_3\varrho^4}(2C_1)^{1+2\alpha}\|\psi_{\Theta}\|_2^2  \\
&  \hspace{40pt} \notag \leq C_9\delta^{-4}(16\delta^{-1}C_1^2Q)^{2\alpha-2}((K^{\frac{2p}{2p-d}}+\delta^{-2})\|\psi_{0,\delta}\|_2^2+\|\zeta_{\Omega}\|_2^2).
\end{align}
Thus,
\begin{equation}\label{q3dp27}
\frac{\alpha^3}{\varrho^4}Q^4((8C_1Q)^{-1}\delta)^{2\alpha+2}\|\psi_{\Theta}\|_2^2\leq C_{10}((K^{\frac{2p}{2p-d}}+\delta^{-2})\|\psi_{0,\delta}\|_2^2+\|\zeta_{\Omega}\|_2^2).
\end{equation}
Since $(\tfrac{\delta}{Q})^5\leq(\tfrac{1}{2})^5\leq\tfrac{1}{8C_1}$ by \eqref{q3d4}, we have
\begin{equation}\label{q3dp28}
\frac{\alpha^3}{\varrho^4}Q^6\left(\frac{\delta}{Q}\right)^{12\alpha+14}\|\psi_{\Theta}\|_2^2\leq C_{11}((1+K^{\frac{2p}{2p-d}})\|\psi_{0,\delta}\|_2^2+\delta^2\|\zeta_{\Omega}\|_2^2).
\end{equation}

To satisfy \eqref{q3dp13} and \eqref{q3dp24}, we choose
\begin{equation}\label{q3dp29}
\alpha=C_{12}(1+K^{\frac{2p}{3p-2d}})\left(Q^{\frac{4p-2d}{3p-2d}}+\log\frac{\|\psi_{\Omega}\|_{2}}{\|\psi_{\Theta}\|_{2}}\right),
\end{equation}
Combining with \eqref{q3dp28}, and recalling $Q\geq1$, we get
\begin{align}\label{q3dp30}
&(1+K^{\frac{2p}{3p-2d}})^3\left(\frac{\delta}{Q}\right)^{C_{13}(1+K^{\frac{2p}{3p-2d}})\left(Q^{\frac{4p-2d}{3p-2d}}+\log\frac{\|\psi_{\Omega}\|_{2}}{\|\psi_{\Theta}\|_{2}}\right)}\|\psi_{\Theta}\|_2^2
\nonumber\\& \qquad  \qquad \le C_{14}((1+K^{\frac{2p}{2p-d}})\|\psi_{0,\delta}\|_2^2+\delta^2\|\zeta_{\Omega}\|_2^2),
\end{align}
and hence
\begin{equation}\label{q3dp31}
\left(\frac{\delta}{Q}\right)^{m_d(1+K^{\frac{2p}{3p-2d}})(Q^{\frac{4p-2d}{3p-2d}}+\log\frac{\|\psi_{\Omega}\|_{2}}{\|\psi_{\Theta}\|_{2}})}\|\psi_{\Theta}\|_{2}^2\leq\|\psi_{x_0,\delta}\|_{2}^2+\delta^2\|\zeta_{\Omega}\|_{2}^2,
\end{equation}
where $m_d>0$ is a constant depending only on $d$. \smallskip

\noindent{(b)}\ $d=2$:  Let  $(\||V^{(2)}|^p\|_{\varphi^{\ast}})^{\frac{1}{p}}\leq K_2$ with  $p\geq2$.  Given  $K_2>0$ and $M>0$, we have
\begin{equation}\label{q2dp4}
\int_{\mathbb{R}^2}\varphi^{\ast}\left(\frac{|W_M^2|}{M^{-\frac{p-2}{2}}K_2^p}\right)dx\leq\int_{\mathbb{R}^2}\varphi^{\ast}\left(\frac{|V^{(2)}|^p}{K_2^p}\right)dx,
\end{equation}
and hence, using $\||V^{(2)}|^p\|_{\varphi^{\ast}}\leq K_2^p$, we get
\begin{equation}\label{q2dp5}
 \|W_M^2\|_{\varphi^{\ast}}\leq M^{-\frac{p-2}{2}}K_2^p.
 \end{equation}
 Using H\"{o}lder's inequality for Orlicz spaces  \eqref{osholder},  and \eqref{q2dp5}, we get
\begin{align}\label{q2dp6}
\int_{\mathbb{R}^2}W_M^2\omega_{\varrho}^{2-2\alpha}\eta^2\psi^2dx&\leq2\|W_M^2\|_{\varphi^{\ast}}\|\omega_{\varrho}^{2-2\alpha}\eta^2\psi^2\|_{\varphi}
\nonumber\\&\leq 2M^{-\frac{p-2}{2}}K_2^p\|\omega_{\varrho}^{2-2\alpha}\eta^2\psi^2\|_{\varphi}.
\end{align}
Using the Sobolev inequality given in  \cite[Theorem~0.1]{AT}, we obtain
\begin{align}\label{q2dp7}
\|\omega_{\varrho}^{2-2\alpha}\eta^2\psi^2\|_{\varphi}& \leq C_4\left(\int_{\mathbb{R}^2}|\omega_{\varrho}^{1-\alpha}\eta\psi|^2dx+\int_{\mathbb{R}^2}|\nabla(\omega_{\varrho}^{1-\alpha}\eta\psi)|^2dx\right)
\\ \notag  &  \le  C_4\int_{\mathbb{R}^2}|\omega_{\varrho}^{1-\alpha}\eta\psi|^2dx+2C_4\int_{\mathbb{R}^2}|\nabla\omega_{\varrho}^{1-\alpha}|^2\eta^2\psi^2dx\\ \notag &  \qquad +2C_4\int_{\mathbb{R}^2}\omega_{\varrho}^{1-2\alpha}|\nabla(\eta\psi)|^2dx.
\end{align}

Combining \eqref{q3dp4a}, \eqref{q2dp6}, \eqref{q2dp7}, and \eqref{q3dp9} with $d=2$, we conclude that
\begin{align}\label{q2dp10}
&\left(\frac{\alpha^3}{3C_3\varrho^4}-4K_1^2-4M-8C_4M^{-\frac{p-2}{2}}K_2^p-16C_4M^{-\frac{p-2}{2}}K_2^p\frac{\alpha^2}{\varrho^2}\right)\int_{\mathbb{R}^2}\omega_{\varrho}^{-1-2\alpha}\eta^2\psi^2dx
\nonumber\\& \hspace{60pt} +\left(\frac{\alpha}{3C_3\varrho^2}-16C_4M^{-\frac{p-2}{2}}K_2^p\right)\int_{\mathbb{R}^2}\omega_{\varrho}^{1-2\alpha}|\nabla(\eta\psi)|^2dx
\nonumber\\ &  \quad \le 4\int_{\supp\nabla\eta}\omega_{\varrho}^{2-2\alpha}|\nabla\eta|^2|\nabla\psi|^2dx
+\int_{\supp\nabla\eta}\omega_{\varrho}^{2-2\alpha}(\Delta\eta)^2\psi^2dx\\  \notag & \hspace{60pt}
+2\int_{\supp\eta}\omega_{\varrho}^{2-2\alpha}\eta^2\zeta^2dx .
\end{align}
Assuming $\alpha\geq\varrho$  and setting $M=K_2^2\alpha^{\frac{4}{p}}\varrho^{-\frac{4}{p}}$, we have
\begin{align}\label{q2dp11}
&4K_1^2+4M+8C_4M^{-\frac{p-2}{2}}K_2^p+16C_4M^{-\frac{p-2}{2}}K_2^p\frac{\alpha^2}{\varrho^2}
\\&\quad\leq4K_1^2+4M+24C_4M^{-\frac{p-2}{2}}K_2^p\frac{\alpha^2}{\varrho^2}
\nonumber\\&\quad=4K_1^2+4K_2^2(1+6C_4)\alpha^{\frac{4}{p}}\varrho^{-\frac{4}{p}}\leq4K^2(1+6C_4)\alpha^{\frac{4}{p}}\varrho^{-\frac{4}{p}}.\nonumber
\end{align}
Taking
\begin{equation}\label{q2dp13}
\alpha\geq C_5(1+K^{\frac{2p}{3p-4}})\varrho^{\frac{4p-4}{3p-4}}\geq C_5(1+K^{\frac{2p}{3p-4}})\varrho^{\frac{4}{3}},
\end{equation}
we can guarantee that  $\alpha>C_2$,
\begin{equation}\label{q2dp12}
\frac{\alpha^3}{3C_3\varrho^4}\geq3(4K^2(1+6C_4)\alpha^{\frac{4}{p}}\varrho^{-\frac{4}{p}}),
\end{equation}
and
\begin{equation}\label{q2dp14}
\frac{\alpha}{3C_3\varrho^2}-16C_4M^{-\frac{p-2}{2}}K_2^p\geq0.
\end{equation}

Using \eqref{ei1} and recalling \eqref{q3d2}, we obtain
\begin{equation}\label{q2dp15}
\int_{\mathbb{R}^2}\omega_{\varrho}^{-1-2\alpha}\eta^2\psi^2dx\geq\left(\frac{\varrho}{Q}\right)^{1+2\alpha}\|\psi_{\Theta}\|_2^2\geq(2C_1)^{1+2\alpha}\|\psi_{\Theta}\|_2^2.
\end{equation}
Combining \eqref{q2dp10}, \eqref{q2dp12}, \eqref{q2dp14} and \eqref{q2dp15}, we conclude that
\begin{align}\label{q2dp16}
&\frac{2\alpha^3}{9C_3\varrho^4}(2C_1)^{1+2\alpha}\|\psi_{\Theta}\|_2^2\leq 4\int_{\supp\nabla\eta}\omega_{\varrho}^{2-2\alpha}|\nabla\eta|^2|\nabla\psi|^2dx
\nonumber\\&  \hspace{40pt}+\int_{\supp\nabla\eta}\omega_{\varrho}^{2-2\alpha}(\Delta\eta)^2\psi^2dx+2\int_{\supp\eta}\omega_{\varrho}^{2-2\alpha}\eta^2\zeta^2dx.
\end{align}

Given $M>0$, we have
\begin{equation}\label{q2dp18}
\int_{\mathbb{R}^2}\varphi^{\ast}\left(\frac{|W_M|}{M^{-\frac{p-1}{2}}K_2^p}\right)dx\leq\int_{\mathbb{R}^2}\varphi^{\ast}\left(\frac{|V^{(2)}|^p}{K_2^p}\right)dx,
\end{equation}
and hence, using $\||V^{(2)}|^p\|_{\varphi^{\ast}}\leq K_2^p$, we get
$
 \|W_M\|_{\varphi^{\ast}}\leq M^{-\frac{p-1}{2}}K_2^p$. Let  $f\in\mathcal{D}(\nabla)$.  Then,
 using \eqref{q3dp17},  H\"{o}lder's inequality for Orlicz spaces  \eqref{osholder},       and the Sobolev inequality in  \cite[Theorem~0.1]{AT}, we get
\begin{equation}\label{q2dp20}
\left|\int_{\mathbb{R}^2}Vf^2dx\right|\leq (K_1+M^{\frac{1}{2}}+2C_4 M^{-\frac{p-1}{2}}K_2^p)\|f\|_2^2+2C_4 M^{-\frac{p-1}{2}}K_2^p\|\nabla f\|_2^2.
\end{equation}
Taking $M=(4C_4K_2^p)^{\frac{2}{p-1}}$ (we can require $C_4\geq1$), we get
\begin{equation}\label{q2dp21}
\left|\int_{\mathbb{R}^2}Vf^2dx\right|\leq4C_4(1+K^{\frac{p}{p-1}})\|f\|_2^2+\tfrac{1}{2}\|\nabla f\|_2^2.
\end{equation}

We have
\begin{align}\label{q2dp22}
&\int_{\{2C_1Q\leq|x|\leq2C_1Q+1\}}\omega_{\varrho}^{2-2\alpha}(4|\nabla\eta|^2|\nabla\psi|^2+(\Delta\eta)^2\psi^2)dx
\\ & \quad\le 64\left(\frac{C_1\varrho}{2C_1Q}\right)^{2\alpha-2}\int_{\{2C_1Q\leq|x|\leq2C_1Q+1\}}(4|\nabla\psi|^2+\psi^2)dx
\nonumber\\ & \quad\le C_6\left(\tfrac{5}{4}C_1\right)^{2\alpha-2}\int_{\{2C_1Q-1\leq|x|\leq2C_1Q+2\}}(\zeta^2+(1+K^{\frac{p}{p-1}})\psi^2)dx
\nonumber\\ & \quad\le C_6\left(\tfrac{5}{4}C_1\right)^{2\alpha-2}(\|\zeta_{\Omega}\|_2^2+(1+K^{\frac{p}{p-1}})\|\psi_{\Omega}\|_2^2),\nonumber
\end{align}
where we used \eqref{q2dp21} and an interior estimate. Similarly,
\begin{align}\label{q2dp23}
&\int_{\{\frac{\delta}{4}\leq|x|\leq\frac{3\delta}{4}\}}\omega_{\varrho}^{2-2\alpha}(4|\nabla\eta|^2|\nabla\psi|^2+(\Delta\eta)^2\psi^2)dx
\\ & \quad\le 1024\delta^{-4}(4\delta^{-1}C_1\varrho)^{2\alpha-2}\int_{\{\frac{\delta}{4}\leq|x|\leq\frac{3\delta}{4}\}}(4|\nabla\psi|^2+\psi^2)dx
\nonumber\\ & \quad\le C_7\delta^{-4}(4\delta^{-1}C_1\varrho)^{2\alpha-2}\int_{\{|x|\leq\delta\}}(\zeta^2+(K^{\frac{p}{p-1}}+\delta^{-2})\psi^2)dx
\nonumber\\ & \quad\le C_7\delta^{-4}(16\delta^{-1}C_1^2Q)^{2\alpha-2}(\|\zeta_{\Omega}\|_2^2+(K^{\frac{p}{p-1}}+\delta^{-2})\|\psi_{0,\delta}\|_2^2).\nonumber
\end{align}
In addition,
\begin{equation}\label{q2dp24}
\int_{\supp\eta}\omega_{\varrho}^{2-2\alpha}\eta^2\zeta^2dx\leq(4\delta^{-1}C_1\varrho)^{2\alpha-2}\|\zeta_{\Omega}\|_2^2\leq(16\delta^{-1}C_1^2Q)^{2\alpha-2}\|\zeta_{\Omega}\|_2^2.
\end{equation}

If we have
\begin{equation}\label{q2dp25}
\frac{\alpha^3}{\varrho^4}\left(\frac{8}{5}\right)^{2\alpha}\|\psi_{\Theta}\|_2^2\geq C_8(1+K^{\frac{p}{p-1}})\|\psi_{\Omega}\|_2^2,
\end{equation}
we obtain
\begin{equation}\label{q2dp26}
C_6\left(\tfrac{5}{4}C_1\right)^{2\alpha-2}(1+K^{\frac{p}{p-1}})\|\psi_{\Omega}\|_2^2\leq\frac{1}{2}\frac{2\alpha^3}{9C_3\varrho^4}(2C_1)^{1+2\alpha}\|\psi_{\Theta}\|_2^2,
\end{equation}
so we conclude that
\begin{align}\label{q2dp27}
&\frac{\alpha^3}{9C_3\varrho^4}(2C_1)^{1+2\alpha}\|\psi_{\Theta}\|_2^2
\\&  \hspace{40pt} \notag \leq C_9\delta^{-4}(16\delta^{-1}C_1^2Q)^{2\alpha-2}((K^{\frac{p}{p-1}}+\delta^{-2})\|\psi_{0,\delta}\|_2^2+\|\zeta_{\Omega}\|_2^2).
\end{align}
Thus,
\begin{equation}\label{q2dp28}
\frac{\alpha^3}{\varrho^4}Q^4((8C_1Q)^{-1}\delta)^{2\alpha+2}\|\psi_{\Theta}\|_2^2\leq C_{10}((K^{\frac{p}{p-1}}+\delta^{-2})\|\psi_{0,\delta}\|_2^2+\|\zeta_{\Omega}\|_2^2).
\end{equation}
Since $(\tfrac{\delta}{Q})^5\leq(\tfrac{1}{2})^5\leq\tfrac{1}{8C_1}$ by \eqref{q3d4}, we have
\begin{equation}\label{q2dp29}
\frac{\alpha^3}{\varrho^4}Q^6\left(\frac{\delta}{Q}\right)^{12\alpha+14}\|\psi_{\Theta}\|_2^2\leq C_{11}((1+K^{\frac{p}{p-1}})\|\psi_{0,\delta}\|_2^2+\delta^2\|\zeta_{\Omega}\|_2^2).
\end{equation}

To satisfy \eqref{q2dp13} and \eqref{q2dp25}, we choose
\begin{equation}\label{q2dp30}
\alpha=C_{12}(1+K^{\frac{2p}{3p-4}})\left(Q^{\frac{4p-4}{3p-4}}+\log\frac{\|\psi_{\Omega}\|_{2}}{\|\psi_{\Theta}\|_{2}}\right),
\end{equation}
Combining with \eqref{q2dp29}, and recalling $Q\geq1$, we get
\begin{align}\label{q2dp31}
&(1+K^{\frac{2p}{3p-4}})^3\left(\frac{\delta}{Q}\right)^{C_{13}(1+K^{\frac{2p}{3p-4}})\left(Q^{\frac{4p-4}{3p-4}}+\log\frac{\|\psi_{\Omega}\|_{2}}{\|\psi_{\Theta}\|_{2}}\right)}\|\psi_{\Theta}\|_2^2
\nonumber\\& \qquad  \qquad \le  C_{14}((1+K^{\frac{p}{p-1}})\|\psi_{0,\delta}\|_2^2+\delta^2\|\zeta_{\Omega}\|_2^2),
\end{align}
and hence there exists $m>0$ such that
\begin{equation}\label{q2dp32}
\left(\frac{\delta}{Q}\right)^{m(1+K^{\frac{2p}{3p-4}})\left(Q^{\frac{4p-4}{3p-4}}+\log\frac{\|\psi_{\Omega}\|_{2}}{\|\psi_{\Theta}\|_{2}}\right)}\|\psi_{\Theta}\|_{2}^2\leq\|\psi_{x_0,\delta}\|_{2}^2+\delta^2\|\zeta_{\Omega}\|_{2}^2.
\end{equation}

If\ $\|V^{(2)}\|_p\leq K_2<\infty$ for some\ $p>2$, we have $(\||V^{(2)}|^{p'}\|_{\varphi^{\ast}})^{\frac{1}{p'}}\leq K_2$ for any $p'\in[2,p)$ since
\begin{equation}
\int_{\mathbb{R}^2}\varphi^{\ast}\left(\frac{|V^{(2)}|^{p'}}{K_2^{p'}}\right)dx\leq\int_{\mathbb{R}^2}\left(\frac{|V^{(2)}|^{p'}}{K_2^{p'}}\right)^{\frac{p}{p'}}dx
\leq\int_{\mathbb{R}^2}\frac{|V^{(2)}|^p}{K_2^p}dx\leq1.
\end{equation}
We conclude that    \eqref{q2dp32}  holds with  $p'$ substituted for $p$. Letting $p' \uparrow p$ we obtain  \eqref{q2dp32}  since $K_2$  is independent of  $p'$. \smallskip

\noindent{(c)}\ $d=1$:   Let $\|V^{(2)}\|_p\leq K_2$ with  $p\geq2$.  Using H\"{o}lder's inequality and \eqref{q3dp5} with $q=2$, we get
\begin{align}\label{q1dp6}
\int_{\mathbb{R}}W_M^2\omega_{\varrho}^{2-2\alpha}\eta^2\psi^2dx&\leq\|W_M\|_2^2\|\omega_{\varrho}^{2-2\alpha}\eta^2\psi^2\|_{\infty}
\leq M^{-\frac{p-2}{2}}K_2^p\|\omega_{\varrho}^{2-2\alpha}\eta^2\psi^2\|_{\infty}.
\end{align}
Applying  Sobolev's inequality, we obtain
\begin{align}\label{q1dp7}
&\|\omega_{\varrho}^{2-2\alpha}\eta^2\psi^2\|_{\infty}\leq \int_{\mathbb{R}}|\omega_{\varrho}^{1-\alpha}\eta\psi|^2dx+\int_{\mathbb{R}}|(\omega_{\varrho}^{1-\alpha}\eta\psi)'|^2dx
\\&\leq \int_{\mathbb{R}}|\omega_{\varrho}^{1-\alpha}\eta\psi|^2dx+2\int_{\mathbb{R}}|(\omega_{\varrho}^{1-\alpha})'|^2\eta^2\psi^2dx
+2\int_{\mathbb{R}}\omega_{\varrho}^{1-2\alpha}|(\eta\psi)'|^2dx.\notag
\end{align}

Combining \eqref{q3dp4a}, \eqref{q1dp6}, \eqref{q1dp7}, and \eqref{q3dp9} with $d=1$, we conclude that
\begin{align}
&\left(\frac{\alpha^3}{3C_3\varrho^4}-4K_1^2-4M-4M^{-\frac{p-2}{2}}K_2^p-8C_4M^{-\frac{p-2}{2}}K_2^p\frac{\alpha^2}{\varrho^2}\right)\int_{\mathbb{R}}\omega_{\varrho}^{-1-2\alpha}\eta^2\psi^2dx
\nonumber\\& \hspace{60pt} +\left(\frac{\alpha}{3C_3\varrho^2}-8M^{-\frac{p-2}{2}}K_2^p\right)\int_{\mathbb{R}}\omega_{\varrho}^{1-2\alpha}|(\eta\psi)'|^2dx \nonumber
 \label{q1dp10}\\ &  \quad \le 4\int_{\supp\eta'}\omega_{\varrho}^{2-2\alpha}|\eta'|^2|\psi'|^2dx
+\int_{\supp\eta'}\omega_{\varrho}^{2-2\alpha}(\eta'')^2\psi^2dx
\\& \hspace{60pt}
+2\int_{\supp\eta}\omega_{\varrho}^{2-2\alpha}\eta^2\zeta^2dx . \nonumber
\end{align}
Assuming $\alpha\geq\varrho$, and setting $M=K_2^2\alpha^{\frac{4}{p}}\varrho^{-\frac{4}{p}}$, we have
\begin{align}\label{q1dp11}
&4K_1^2+4M+4M^{-\frac{p-2}{2}}K_2^p+8M^{-\frac{p-2}{2}}K_2^p\frac{\alpha^2}{\varrho^2}
\\&\quad\leq4K_1^2+4M+12M^{-\frac{p-2}{2}}K_2^p\frac{\alpha^2}{\varrho^2}
=4K_1^2+16K_2^2\alpha^{\frac{4}{p}}\varrho^{-\frac{4}{p}}\leq16K^2\alpha^{\frac{4}{p}}\varrho^{-\frac{4}{p}}. \nonumber
\end{align}
Taking
\begin{equation}\label{q1dp13}
\alpha\geq C_5(1+K^{\frac{2p}{3p-4}})\varrho^{\frac{4p-4}{3p-4}}\geq C_5(1+K^{\frac{2p}{3p-4}})\varrho^{\frac{4}{3}},
\end{equation}
we can guarantee that  $\alpha>C_2$,
\begin{equation}\label{q1dp12}
\frac{\alpha^3}{3C_3\varrho^4}\geq3(16K^2\alpha^{\frac{4}{p}}\varrho^{-\frac{4}{p}}),
\end{equation}
and
\begin{equation}\label{q1dp14}
\frac{\alpha}{3C_3\varrho^2}-8M^{-\frac{p-2}{2}}K_2^p\geq0.
\end{equation}

Using \eqref{ei1} and recalling \eqref{q3d2}, we obtain
\begin{equation}\label{q1dp15}
\int_{\mathbb{R}}\omega_{\varrho}^{-1-2\alpha}\eta^2\psi^2dx\geq\left(\frac{\varrho}{Q}\right)^{1+2\alpha}\|\psi_{\Theta}\|_2^2\geq(2C_1)^{1+2\alpha}\|\psi_{\Theta}\|_2^2.
\end{equation}
Combining \eqref{q1dp10}, \eqref{q1dp12}, \eqref{q1dp14} and \eqref{q1dp15}, we conclude that
\begin{align}\label{q1dp16}
&\frac{2\alpha^3}{9C_3\varrho^4}(2C_1)^{1+2\alpha}\|\psi_{\Theta}\|_2^2\leq4\int_{\supp\eta'}\omega_{\varrho}^{2-2\alpha}|\eta'|^2|\psi'|^2dx
\nonumber\\&  \hspace{40pt}+\int_{\supp\eta'}\omega_{\varrho}^{2-2\alpha}(\eta'')^2\psi^2dx+2\int_{\supp\eta}\omega_{\varrho}^{2-2\alpha}\eta^2\zeta^2dx
\end{align}

Let $f\in\mathcal{D}(\nabla)$ and $M>0$.
Using \eqref{q3dp17}, H\"{o}lder's inequality, \eqref{q3dp5} with $d=1$,
and Sobolev's inequality, we get
\begin{equation}\label{q1dp20}
\left|\int_{\mathbb{R}}Vf^2dx\right|\leq (K_1+M^{\frac{1}{2}}+M^{-\frac{p-1}{2}}K_2^p)\|f\|_2^2+ M^{-\frac{p-1}{2}}K_2^p\|f'\|_2^2.
\end{equation}
Taking\ $M=(2K_2^p)^{\frac{2}{p-1}}$, we get
\begin{equation}\label{q1dp21}
\left|\int_{\mathbb{R}}Vf^2dx\right|\leq2(1+K^{\frac{p}{p-1}})\|f\|_2^2+\frac{1}{2}\|f'\|_2^2.
\end{equation}

We have
\begin{align}\label{q1dp22}
&\int_{\{2C_1Q\leq|x|\leq2C_1Q+1\}}\omega_{\varrho}^{2-2\alpha}(4|\eta'|^2|\psi'|^2+(\eta'')^2\psi^2)dx
\\ & \quad\le 64\left(\frac{C_1\varrho}{2C_1Q}\right)^{2\alpha-2}\int_{\{2C_1Q\leq|x|\leq2C_1Q+1\}}(4|\psi'|^2+\psi^2)dx
\nonumber\\ & \quad\le C_6\left(\frac{5}{4}C_1\right)^{2\alpha-2}\int_{\{2C_1Q-1\leq|x|\leq2C_1Q+2\}}(\zeta^2+(1+K^{\frac{p}{p-1}})\psi^2)dx
\nonumber\\ & \quad\le C_6\left(\frac{5}{4}C_1\right)^{2\alpha-2}(\|\zeta_{\Omega}\|_2^2+(1+K^{\frac{p}{p-1}})\|\psi_{\Omega}\|_2^2), \notag
\end{align}
where we used \eqref{q2dp21} and an interior estimate. Similarly,
\begin{align}\label{q1dp23}
&\int_{\{\frac{\delta}{4}\leq|x|\leq\frac{3\delta}{4}\}}\omega_{\varrho}^{2-2\alpha}(4|\eta'|^2|\psi'|^2+(\eta'')^2\psi^2)dx
\\ & \quad\le 1024\delta^{-4}(4\delta^{-1}C_1\varrho)^{2\alpha-2}\int_{\{\frac{\delta}{4}\leq|x|\leq\frac{3\delta}{4}\}}(4|\psi'|^2+\psi^2)dx
\nonumber\\ & \quad\le C_7\delta^{-4}(4\delta^{-1}C_1\varrho)^{2\alpha-2}\int_{\{|x|\leq\delta\}}(\zeta^2+(K^{\frac{p}{p-1}}+\delta^{-2})\psi^2)dx
\nonumber\\ & \quad\le C_7\delta^{-4}(16\delta^{-1}C_1^2Q)^{2\alpha-2}(\|\zeta_{\Omega}\|_2^2+(K^{\frac{p}{p-1}}+\delta^{-2})\|\psi_{0,\delta}\|_2^2). \notag
\end{align}
In addition,
\begin{equation}\label{q1dp24}
\int_{\supp\eta}\omega_{\varrho}^{2-2\alpha}\eta^2\zeta^2dx\leq(4\delta^{-1}C_1\varrho)^{2\alpha-2}\|\zeta_{\Omega}\|_2^2\leq(16\delta^{-1}C_1^2Q)^{2\alpha-2}\|\zeta_{\Omega}\|_2^2.
\end{equation}

If we have
\begin{equation}\label{q1dp25}
\frac{\alpha^3}{\varrho^4}\left(\frac{8}{5}\right)^{2\alpha}\|\psi_{\Theta}\|_2^2\geq C_8(1+K^{\frac{p}{p-1}})\|\psi_{\Omega}\|_2^2,
\end{equation}
we obtain
\begin{equation}\label{q1dp26}
C_6\left(\tfrac{5}{4}C_1\right)^{2\alpha-2}(1+K^{\frac{p}{p-1}})\|\psi_{\Omega}\|_2^2\leq\frac{1}{2}\frac{2\alpha^3}{9C_3\varrho^4}(2C_1)^{1+2\alpha}\|\psi_{\Theta}\|_2^2,
\end{equation}
so we conclude that
\begin{align}\label{q1dp27}
&\frac{\alpha^3}{9C_3\varrho^4}(2C_1)^{1+2\alpha}\|\psi_{\Theta}\|_2^2
\\&  \hspace{40pt} \notag\leq C_9\delta^{-4}(16\delta^{-1}C_1^2Q)^{2\alpha-2}((K^{\frac{p}{p-1}}+\delta^{-2})\|\psi_{0,\delta}\|_2^2+\|\zeta_{\Omega}\|_2^2).
\end{align}
Thus,
\begin{equation}\label{q1dp28}
\frac{\alpha^3}{\varrho^4}Q^4((8C_1Q)^{-1}\delta)^{2\alpha+2}\|\psi_{\Theta}\|_2^2\leq C_{10}((K^{\frac{p}{p-1}}+\delta^{-2})\|\psi_{0,\delta}\|_2^2+\|\zeta_{\Omega}\|_2^2).
\end{equation}
Since $(\tfrac{\delta}{Q})^5\leq(\tfrac{1}{2})^5\leq\tfrac{1}{8C_1}$ by \eqref{q3d4}, we have
\begin{equation}\label{q1dp29}
\frac{\alpha^3}{\varrho^4}Q^6\left(\frac{\delta}{Q}\right)^{12\alpha+14}\|\psi_{\Theta}\|_2^2\leq C_{11}((1+K^{\frac{p}{p-1}})\|\psi_{0,\delta}\|_2^2+\delta^2\|\zeta_{\Omega}\|_2^2).
\end{equation}

To satisfy \eqref{q1dp13} and \eqref{q1dp25}, we choose
\begin{equation}\label{q1dp30}
\alpha=C_{12}(1+K^{\frac{2p}{3p-4}})\left(Q^{\frac{4p-4}{3p-4}}+\log\frac{\|\psi_{\Omega}\|_{2}}{\|\psi_{\Theta}\|_{2}}\right),
\end{equation}
Combining with\ \eqref{q1dp29}, and recalling\ $Q\geq1$, we get
\begin{align}\label{q1dp31}
&(1+K^{\frac{2p}{3p-4}})^3\left(\frac{\delta}{Q}\right)^{C_{13}(1+K^{\frac{2p}{3p-4}})\left(Q^{\frac{4p-4}{3p-4}}+\log\frac{\|\psi_{\Omega}\|_{2}}{\|\psi_{\Theta}\|_{2}}\right)}\|\psi_{\Theta}\|_2^2
\\& \qquad  \qquad \le C_{14}((1+K^{\frac{p}{p-1}})\|\psi_{0,\delta}\|_2^2+\delta^2\|\zeta_{\Omega}\|_2^2), \nonumber
\end{align}
and hence there exists\ $m>0$\ such that
\begin{equation}\label{q1dp32}
\left(\frac{\delta}{Q}\right)^{m(1+K^{\frac{2p}{3p-4}})\left(Q^{\frac{4p-4}{3p-4}}+\log\frac{\|\psi_{\Omega}\|_{2}}{\|\psi_{\Theta}\|_{2}}\right)}\|\psi_{\Theta}\|_{2}^2\leq\|\psi_{x_0,\delta}\|_{2}^2+\delta^2\|\zeta_{\Omega}\|_{2}^2.
\end{equation}
\end{proof}

\section{Unique continuation principle for spectral projections of Schr\"{o}dinger operators with singular potentials}\label{secUCPSP}

The following theorem, a consequence of  Theorem~\ref{qucp3d},  is an extension of \cite[Theorem~B.4]{KN} to Schr\"{o}dinger operators with singular potentials. Theorem~\ref{ucpsp} follows   from  Theorem~\ref{ucpt}.

\begin{theorem}\label{ucpt}
Let $H=-\Delta+V$ be a Schr\"{o}dinger operator on $\mathrm{L}^2(\mathbb{R}^d)$, where $V=V^{(1)}+V^{(2)}$ with $\|V^{(1)}\|_{\infty}\leq K_1<\infty$ and $\|V^{(2)}\|_p\leq K_2<\infty$ with $p\geq d$ for $d\geq3$, $p>2$ for $d=2$, and $p\geq2$ for $d=1$.   Set $K=K_1+K_2$. Fix $\delta\in(0,\frac{1}{2}]$, let $\{y_k\}_{k\in\mathbb{Z}^d}$ be sites in $\mathbb{R}^d$ with $B(y_k,\delta)\subset\Lambda_1(k)$ for all $k\in\mathbb{Z}^d$.
There exists a constant $M_d>0$, such that given a rectangle $\Lambda$ as in \eqref{ucp2}, where $a\in\mathbb{R}^d$ and $L_j\geq114\sqrt{d}$ for $j=1,\ldots,d$, and a real-valued $\psi\in\mathcal{D}(H_{\Lambda})$, we have
\begin{equation}\label{ucpt1}
\delta^{M_d\left(1+K^{\beta_{d,p}}\right)}\|\psi_{\Lambda}\|_2^2\leq\sum_{k\in \mathbb{Z}^d, \, \Lambda_1(k)\subset\Lambda}\|\psi_{y_k,\delta}\|_2^2+\delta^2\|((-\Delta+V)\psi)_{\Lambda}\|_2^2,
\end{equation}
where
\begin{equation}\label{ucpt2}
\beta_{d,p}=\left\{\begin{array}{ll}
\frac{2p}{3p-2d}&\mbox{for}\quad d\geq2\\\frac{2p}{3p-4}&\mbox{for}\quad d=1
\end{array}
\right..
\end{equation}
\end{theorem}

\begin{proof}[Proof of Theorem~\ref{ucpt}]  Under  the hypotheses of the theorem   $V\in \mathrm{L}_{loc}^2(\mathbb{R}^d)$, which implies that    $\mathcal{D}(\Delta_{\Lambda})\cap \{\phi \in \mathrm{L}^2(\Lambda): \, V\phi \in\mathrm{L}^2(\Lambda) \}$ is an operator core for $H_{\Lambda}$, so  it suffices to prove the theorem for  $\psi \in \mathcal{D}(\Delta_{\Lambda})$ with  $V\psi \in\mathrm{L}^2(\Lambda)$.

Using the notation in the proof of \cite[Theorem B.4]{KN}, we have $\|\widehat{V^{(1)}}\|_{\infty}=\|V^{(1)}\|_{\infty}\leq K_1$ and $\|\widehat{V^{(2)}}_{\Lambda_{Y\mathbf{\tau}}(\kappa)}\|_p\leq 3^d\|V^{(2)}_{\Lambda}\|_p\leq 3^dK_2$ for any $\kappa\in\Lambda$, since $\Lambda_{Y\mathbf{\tau}}(\kappa)\subset\Lambda_{3\mathbf{L}}$ as $Y\tau_j<\frac{L_j}{2}, j=1,2,\ldots,d$. Using Theorem \ref{qucp3d} and following the proof of \cite[Theorem B.4]{KN}, we prove \eqref{ucpt1}.
\end{proof}

\begin{proof}[Proof of Theorem~\ref{ucpsp}]
From \eqref{q3dp20}, \eqref{q2dp21} and \eqref{q1dp21}, there exists a constant $C_d>0$
such that for all\ $f\in\mathcal{D}(\nabla)$
\begin{equation}\label{utp1}
\left|\int_{\mathbb{R}^d}Vf^2dx\right|\leq\theta\|f\|_2^2+\frac{1}{2}\|\nabla f\|_2^2
\end{equation}
where $\theta=C_d(1+K^{\frac{2p}{2p-d}})$ for $d\geq2$ and $\theta=C_1(1+K^{\frac{p}{p-1}})$ for $d=1$.
Therefore $\sigma(H_{\Lambda})\subset[-\theta,\infty)$, and hence it suffices to consider\ $E_0\geq-\theta$\ and\ $E\in[-\theta,E_0]$. We have $V-E=(V^{(1)}-E)+V^{(2)}$, where
\begin{equation}\label{utp2}
\|V^{(1)}-E\|_{\infty}\leq\|V^{(1)}\|_{\infty}+\max\{E_0,\theta\}\leq K_1+E_0+\theta
\end{equation}
and $\|V^{(2)}\|_{p}\leq K_2$. Applying Theorem \ref{ucpt} and following the proof of \cite[Theorem B.1]{KN}, we prove \eqref{ucpsp3}.
\end{proof}

\end{document}